\documentclass[11pt]{article}
\usepackage{amsmath,amssymb,amsthm}
\usepackage{graphicx}
\usepackage{booktabs}
\usepackage{geometry}
\usepackage{hyperref}
\usepackage{comment}
\usepackage{multirow}
\usepackage{xcolor}
\newtheorem{theorem}{Theorem}
\newtheorem{lemma}{Lemma}

\newtheorem{remark}{Remark}

\title{\textbf{Robust M-Estimation for Spliced Power-Law Models Under Cusp-Type Non-Regularity}}

\author{
  \textbf{N. Balakrishna} \\
  Department of Mathematics and Statistics \\
  Indian Institute of Technology Tirupati \\
  Tirupati, Andhra Pradesh, India \\
  \texttt{bala@iittp.ac.in}
}

\date{}

\begin{document}

\maketitle

\begin{abstract}
Statistical modeling of heavy-tailed data with bounded or finite body behavior often relies on composite or spliced distributions characterized by structural thresholds. Because the probability density function features a sharp directional shift (a cusp) at the join point, the classical log-likelihood function exhibits a discontinuity in its first derivative, causing standard Cram\'er-Rao regularity conditions and maximum likelihood frameworks to break down. To overcome this non-regularity, we develop a robust, log-transformed asymmetric Huber M-estimation framework for the structural threshold, coupled with conditional shape maximum likelihood estimators. We derive a closed-form parameter-balancing condition that analytically guarantees Fisher consistency under structural tail asymmetry. Furthermore, we establish the joint asymptotic normality of the hybrid Maximum Likelihood and M-estimator system and derive its explicit sandwich covariance structure. A computationally efficient profiled Z-estimation algorithm is implemented to reduce the joint estimation to a one-dimensional root-finding problem. Finite-sample properties are validated through Monte Carlo simulations, and the practical utility of the framework is demonstrated on the benchmark Danish Fire Insurance claims dataset.
\end{abstract}



\section{Introduction}
Statistical modeling of heavy-tailed data displaying bounded or finite body behavior is critical across quantitative finance, insurance, and risk assessment. While standard unimodal parametric distributions frequently struggle to fit empirical datasets that exhibit two distinct regimes: a high-density, bounded body representing common lower-magnitude events, and a slow-decaying power-law tail representing extreme losses, composite (or spliced) distributions have gained significant attention. Splicing allows a continuous transition between two or more probability density functions across predefined join points. In this work, we investigate a Two-Piece Double Pareto-Like Distribution, formed as a spliced mixture combining a finite-range power function distribution (modeling typical, operational data) with a Pareto Type-I distribution (modeling extreme right-tail risk).

This structure naturally fits several practical applications: For example, (i) In Actuarial Science \& Insurance Claims, small-to-moderate claims occur within a bounded range dictated by policy deductibles or administrative caps, whereas catastrophic claims (e.g., major flood or fire losses) extend far into the upper tail. (ii) In Operational Risk Management, financial institutions face high-frequency, low-impact losses alongside low-frequency, high-impact events that dictate Basel capital requirements. (iii) In Reinsurance \& Layered Risk Transfer, coverage is structured in layers where primary coverage and excess-of-loss reinsurers face different probabilistic laws.
Our proposed distribution sits at the intersection of two-piece scale/shape distribution theory and robust non-smooth optimization.

Because the probability density function features a sharp directional shift (a cusp) at the threshold $x=m$, the classical log-likelihood function exhibits a discontinuity in its first derivative with respect to $m$. Consequently, standard Cram\'er–Rao regularity frameworks collapse for the location parameter, and classical MLE properties default to order-statistic bounds. By constructing a log-transformed Huber M-estimator using $z = \ln(X) - \ln(m)$, we resolve this non-regularity via bounded, convex optimization scores.

The key novelties of this paper are twofold: (1) We establish explicit, closed-form asymptotic dispersion matrices $\Sigma(\eta_{0})$ for the shape vector $\eta = (\theta, \beta)^{T}$ and $\Omega$ for the full parameter vector $\Xi = (\theta, \beta, m)^{T}$. (2) We demonstrate that enforcing an exact tuning constraint, $\theta(1-e^{-\beta k}) = \beta(1-e^{-\theta k})$, analytically nullifies structural bias, preserving Fisher consistency ($\mathbb{E}[\psi_{k}(Z)]=0$) under an asymmetric generalized Huber framework.

The rest of the paper is organized as follows: Section 2 formulates the two-piece double Pareto-like distribution and derives its structural properties. Section 3 develops the hybrid parameter estimation framework, deriving the conditional shape MLEs, the asymmetric Huber threshold estimator, the joint sandwich covariance structure and the profiled Z-estimation procedure. Section 4 evaluates finite-sample performance via Monte Carlo simulations. Section 5 applies the methodology to the benchmark Danish Fire Insurance claims dataset. Section 6 concludes the paper. Detailed proofs and verifications regularity conditions are deferred to the Appendix.

\section{Formulation and Properties of the Distribution}
Let $X$ be a non-negative continuous random variable. Consider a two-component spliced density function composed of a power function distribution $f_{1}(x)$ on $[0, m)$ and a Pareto Type-I distribution $f_{2}(x)$ on $[m, \infty)$:
\begin{align*}
f_{1}(x; \beta, m) &= \left(\frac{\beta}{m}\right)\left(\frac{x}{m}\right)^{\beta-1}, \quad 0 \le x < m, \quad \beta > 0, m > 0, \\
f_{2}(x; \theta, m) &= \left(\frac{\theta}{m}\right)\left(\frac{m}{x}\right)^{\theta+1}, \quad x \ge m, \quad \theta > 0, m > 0.
\end{align*}
Combining these with weight parameter $p \in (0,1)$ yields the spliced density $f(x) = p f_{1}(x) \mathbb{I}(0 \le x < m) + (1-p) f_{2}(x) \mathbb{I}(x \ge m)$. Imposing the continuity constraint $f(m^{-}) = f(m^{+})$ yields:
\[
p \cdot \frac{\beta}{m} = (1-p) \cdot \frac{\theta}{m} \implies p = \frac{\theta}{\theta + \beta}.
\]
Substituting $p$ and $1-p = \frac{\beta}{\theta + \beta}$ back into the piecewise density eliminates $p$, yielding the continuous three-parameter Two-Piece Double Pareto-Like Probability Density Function (PDF):
\begin{equation}\label{DPpdf}
f(x; \theta, \beta, m) = \begin{cases}
\frac{\beta \theta}{\theta + \beta} \cdot \frac{x^{\beta-1}}{m^{\beta}}, & \text{if } 0 < x \le m, \\
\frac{\theta \beta}{\theta + \beta} \cdot \frac{m^{\theta}}{x^{\theta+1}}, & \text{if } x > m.
\end{cases}
\end{equation}
The corresponding Cumulative Distribution Function (CDF) is:
\begin{equation}\label{DPcdf}
F(x; \theta, \beta, m) = \begin{cases}
\frac{\theta}{\theta + \beta} \left(\frac{x}{m}\right)^{\beta}, & \text{if } 0 < x \le m, \\
1 - \frac{\beta}{\theta + \beta} \left(\frac{m}{x}\right)^{\theta}, & \text{if } x > m.
\end{cases}
\end{equation}
To illustrate the theoretical properties of the double Pareto distribution under varying parameter regimes, Figure 1 displays side-by-side plots of the PDF, $f(x; \theta, \beta, m)$ and the corresponding CDF, $F(x; \theta, \beta, m)$.

\begin{figure}[htbp]
    \centering
    \includegraphics[width=1\linewidth]{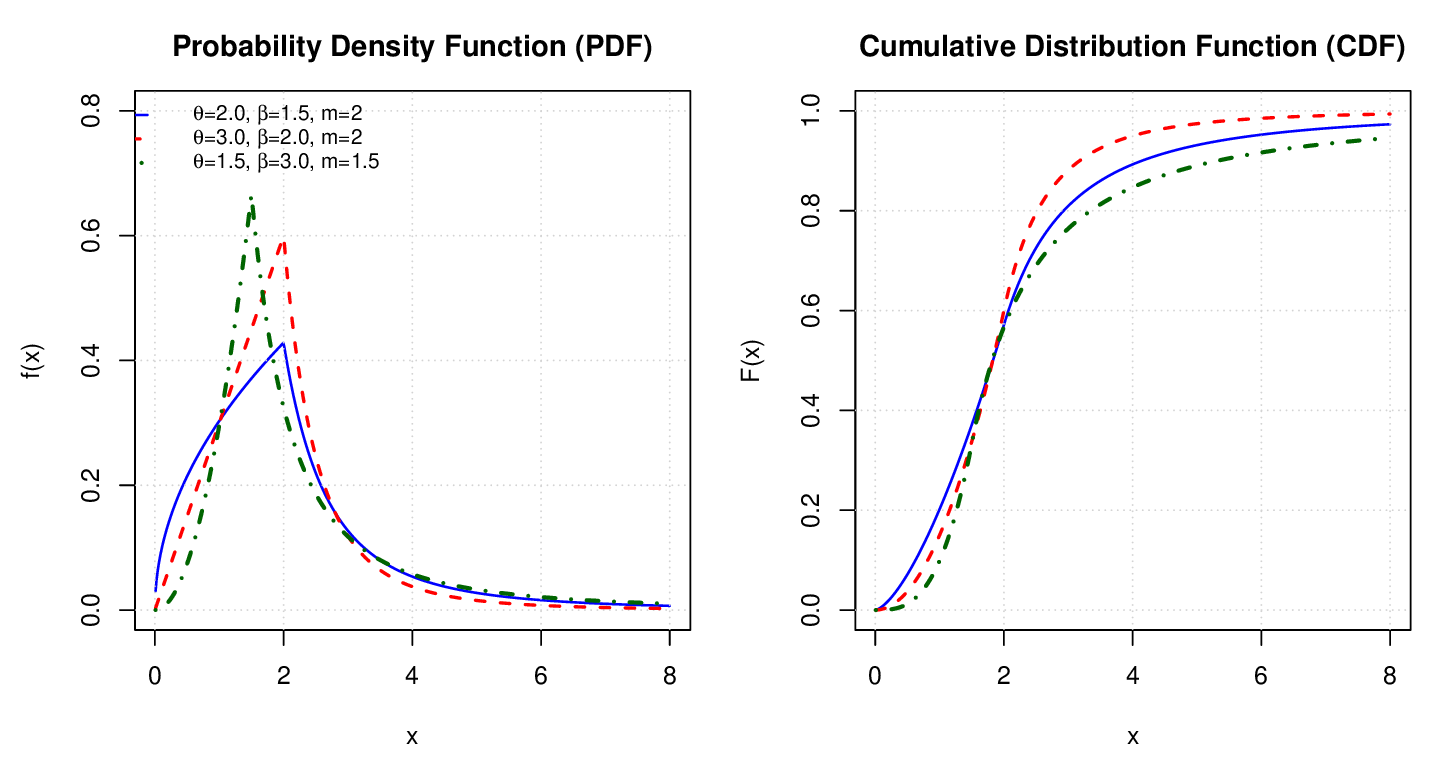}
    \caption{PDF and the corresponding CDF plots of Double Pareto distribution}
\end{figure}

The threshold parameter $m$ marks the structural pivot point where the distribution transitions between its two power-law regimes, corresponding to the peak (mode) of $f(x)$. Below $m$ ($0 < x \le m$), the lower shape parameter $\beta$ dictates the rate of growth: higher values of $\beta$ push density toward $m$, sharpening the left-tail rise and accelerating early cumulative probability in $F(x)$. Above $m$ ($x > m$), the upper shape parameter $\theta$ controls the heavy-tailed decay; larger values of $\theta$ induce a steeper drop in $f(x)$, causing $F(x)$ to approach unity more rapidly, whereas smaller values yield a heavier right tail.

The $k$-th raw moment exists if and only if $\theta > k$, and is given by:
\[
\mathbb{E}[X^{k}] = \frac{\beta \theta m^{k}}{(\beta + k)(\theta - k)}, \quad \theta > k.
\]
For $\theta > 1$ and $\theta > 2$, the mean $\mu_{X}$ and variance $\sigma_{X}^{2}$ are:
\begin{align*}
\mu_{X} &= \frac{m \theta \beta}{(\beta+1)(\theta-1)}, \\
\sigma_{X}^{2} &= \frac{m^{2} \theta \beta [\beta^{2} + 2\beta + 1 + \theta^{2} - 2\theta]}{(\beta+2)(\theta-2)(\beta+1)^{2}(\theta-1)^{2}}.
\end{align*}
Inverting $F(x)$ yields the quantile function $x(u) = F^{-1}(u)$ for $u \in (0,1)$:
\begin{equation}\label{quantile}
x(u) = \begin{cases}
m \left(\frac{u}{p}\right)^{1/\beta}, & \text{if } u \le p, \\
m \left(\frac{1-p}{1-u}\right)^{1/\theta}, & \text{if } u > p,
\end{cases}
\end{equation}
where $p = \frac{\theta}{\theta + \beta} = F(m)$ is the cumulative probability mass below the structural threshold.

\section{Parameter Estimation}
In this section, we address the inferential challenges arising from the non-regularity of the distribution. Because of the cusp at $m$, standard Cramér–Rao regularity conditions fail. Conditional on a fixed threshold $m$, we first derive closed-form maximum likelihood estimators for $\eta = (\theta, \beta)^T$. Next, we introduce a robust log-transformed asymmetric Huber M-estimator for $m$ with an exact parameter-balancing condition ensuring Fisher consistency. Finally, we integrate both estimating mechanisms into a simultaneous Z-estimation system to establish the joint asymptotic normality of $\Xi = (\theta, \beta, m)^T$.

\subsection{Maximum Likelihood Estimation for Shape Parameters (\texorpdfstring{$\theta, \beta$}{theta, beta})}
Let $X = (X_{1}, X_{2}, \dots, X_{n})$ be an independent and identically distributed (i.i.d.) sample from $f(x; \theta, \beta, m)$. For a fixed $m$, let $I_{i}(m) = \mathbb{I}(0 < X_{i} \le m)$. The log-likelihood function $l(\theta, \beta, m; X)$ is:
\[
l(\theta, \beta, m; X) = n \ln \theta + n \ln \beta - n \ln(\theta + \beta) + (\beta - 1) \sum_{i=1}^{n} I_{i}(m) \ln\left(\frac{X_{i}}{m}\right) - (\theta + 1) \sum_{i=1}^{n} (1 - I_{i}(m)) \ln\left(\frac{X_{i}}{m}\right).
\]
Setting partial derivatives with respect to $\theta$ and $\beta$ to zero yields:
\begin{align*}
\frac{\partial l}{\partial \theta} &= \frac{n}{\theta} - \frac{n}{\theta + \beta} - \sum_{i=1}^{n} (1 - I_{i}(m)) \ln\left(\frac{X_{i}}{m}\right) = 0, \\
\frac{\partial l}{\partial \beta} &= \frac{n}{\beta} - \frac{n}{\theta + \beta} + \sum_{i=1}^{n} I_{i}(m) \ln\left(\frac{X_{i}}{m}\right) = 0.
\end{align*}
Define the sample tail-log statistics as $$T_{1}(m) = \sum_{i=1}^{n} I_{i}(m) \ln(m/X_{i}) > 0 \ \text{and} \    T_{2}(m) = \sum_{i=1}^{n} (1 - I_{i}(m)) \ln(X_{i}/m) > 0.$$ Then:
\[
\frac{\beta}{\theta(\theta + \beta)} = \frac{T_{2}(m)}{n}, \quad \frac{\theta}{\beta(\theta + \beta)} = \frac{T_{1}(m)}{n}.
\]
Dividing these equations gives $\frac{\beta}{\theta} = \sqrt{\frac{T_{2}(m)}{T_{1}(m)}}$. Solving this system yields the closed-form conditional MLEs:
\begin{equation}\label{eq:mle_closed}
\hat{\theta}_{\text{MLE}}(m) = \frac{n}{T_{2}(m) + \sqrt{T_{1}(m)T_{2}(m)}}, \quad \hat{\beta}_{\text{MLE}}(m) = \frac{n}{T_{1}(m) + \sqrt{T_{1}(m)T_{2}(m)}}.
\end{equation}
For a fixed threshold $m$, the parameter vector $\eta = (\theta, \beta)^{T}=(\eta_1, \eta_2)^{T}$ (say), belongs to the open parameter space $\Theta = (0, \infty) \times (0, \infty) \subset \mathbb{R}^{2}$. In Appendix A.1, we verify that $f(x; \eta, m)$ satisfies the multi-parameter Cram\'er-type regularity conditions established by Chanda (1954). Consequently, by the main theorem of Chanda (1954), there exists a consistent root $\hat{\eta}_{\text{MLE}} = (\hat{\theta}_{\text{MLE}}, \hat{\beta}_{\text{MLE}})^{T}$ of the likelihood equations that represents a local maximum of the likelihood function with probability approaching one, and satisfies:
\begin{equation}
\sqrt{n}(\hat{\eta}_{\text{MLE}} - \eta_{0}) \xrightarrow{d} \mathcal{N}_{2}(0, \Sigma(\eta_{0})),
\end{equation}
where $\Sigma(\eta_{0}) = J(\eta_{0})^{-1}$ and $J(\eta) = (J_{jk}(\eta))_{2 \times 2}$ is the per-observation Fisher Information matrix:
\begin{equation}
J_{jk}(\eta) = \mathbb{E}\left[ \frac{\partial \ln f(X; \eta, m)}{\partial \eta_j} \frac{\partial \ln f(X; \eta, m)}{\partial \eta_k} \right] = -\mathbb{E}\left[ \frac{\partial^2 \ln f(X; \eta, m)}{\partial \eta_j \partial \eta_k} \right], j,k=1,2.
\end{equation}
Evaluating the second-order partial derivatives of the log-density gives:
\begin{equation}
\frac{\partial^{2} \ln f}{\partial \theta^{2}} = -\frac{1}{\theta^{2}} + \frac{1}{(\theta + \beta)^{2}}, \quad \frac{\partial^{2} \ln f}{\partial \beta^{2}} = -\frac{1}{\beta^{2}} + \frac{1}{(\theta + \beta)^{2}}, \quad \frac{\partial^{2} \ln f}{\partial \theta \partial \beta} = \frac{1}{(\theta + \beta)^{2}},
\end{equation}
which are all deterministic (independent of $X$). Hence,
\begin{equation}\label{eq:fisher_matrix}
J(\eta) = \begin{pmatrix}
\frac{\beta(\beta + 2\theta)}{\theta^{2}(\theta + \beta)^{2}} & -\frac{1}{(\theta + \beta)^{2}} \\
-\frac{1}{(\theta + \beta)^{2}} & \frac{\theta(\theta + 2\beta)}{\beta^{2}(\theta + \beta)^{2}}
\end{pmatrix}.
\end{equation}
The matrix determinant is $\det J(\eta) = \frac{2}{\theta \beta (\theta + \beta)^{2}} > 0$. Inverting $J(\eta_0)$ gives the explicit asymptotic dispersion matrix:
\begin{equation}\label{eq:sigma_eta}
\Sigma(\eta_{0}) = \begin{pmatrix}
\frac{\theta_{0}^{3} + 2\theta_{0}^{2}\beta_{0}}{2\beta_{0}} & \frac{\theta_{0}\beta_{0}}{2} \\
\frac{\theta_{0}\beta_{0}}{2} & \frac{\beta_{0}^{3} + 2\beta_{0}^{2}\theta_{0}}{2\theta_{0}}
\end{pmatrix}.
\end{equation}
Note that the conditional MLE (\ref{eq:mle_closed}) depends on $m$, which must be estimated through a suitable robust framework as developed in the following subsection.

\subsection{Robust M-Estimation for the Structural Threshold \texorpdfstring{$m$}{m}}
Because $\frac{\partial}{\partial m} \ln f(x; \theta, \beta, m)$ undergoes a discontinuous step jump at $x = m$, classical likelihood estimation for $m$ breaks down. To address this non-regularity, we propose Huber's robust M-estimation framework applied to the log-transformed residuals $z_{i} = \ln(X_{i}/m)$. We introduce an asymmetric Huber score function $\psi_{k}(z)$ with tuning parameter $k > 0$:
\begin{equation}\label{psik}
\psi_{k}(z) = \begin{cases}
\beta, & \text{if } z < -k, \\
-\frac{\beta}{k}z, & \text{if } -k \le z < 0, \\
-\frac{\theta}{k}z, & \text{if } 0 \le z \le k, \\
-\theta, & \text{if } z > k.
\end{cases}
\end{equation}

The M-estimator $\hat{\mu}_{M} = \ln \hat{m}_{M}$ solves the estimating equation:
\begin{equation}\label{eq:m_est_eq}
\sum_{i=1}^{n} \psi_{k}(Y_{i} - \mu) = 0, \quad \text{where}\quad Y_i = \ln X_i.
\end{equation}
This score function replaces the discontinuous step jump at $z = 0$ with continuous linear segments, yielding a continuously differentiable criterion while clipping extreme residuals beyond $[-k, k]$ to ensure bounded gross-error sensitivity.

\begin{theorem}\label{k-value}
The M-estimator score satisfies Fisher consistency, $\mathbb{E}_{\mu_{0}}[\psi_{k}(Z)] = 0$, if and only if the tuning boundary $k$ satisfies:
\begin{equation}\label{balance}
\theta(1 - e^{-\beta k}) = \beta(1 - e^{-\theta k}).
\end{equation}
\end{theorem}

\begin{remark}\label{Appr-k}
Expanding $1-e^{-x} \approx x - \frac{x^2}{2}+\frac{x^3}{6}$ around $k \approx 0$ yields the explicit balancing approximation:
\begin{equation}\label{eq:k_approx}
k \approx \frac{3}{\beta+\theta}.
\end{equation}
\end{remark}

\begin{theorem}\label{mhat-CAN}
As $n \to \infty$, the robust threshold estimator $\hat{m}_{M} = e^{\hat{\mu}_{M}}$ satisfies:
\[
\sqrt{n}(\hat{m}_{M} - m_{0}) \xrightarrow{d} \mathcal{N}(0, \sigma_{M}^{2}), \quad \sigma_{M}^{2} = m_{0}^{2} \frac{W}{V^{2}},
\]
where $V = -\mathbb{E}[\psi_{k}'(Z)] > 0$ is:
\begin{equation}\label{eq:V_def}
V = \frac{\theta \beta}{k(\theta + \beta)} \left[(1 - e^{-\beta k}) + (1 - e^{-\theta k})\right],
\end{equation}
and $W = \mathbb{E}[\psi_{k}^{2}(Z)] < \infty$ is given explicitly by:
\begin{equation}\label{eq:W_def}
W = \frac{2}{k^{2}} - \frac{2\theta\beta}{k(\theta + \beta)} (e^{-\beta k} + e^{-\theta k}) - \frac{2}{k^{2}(\theta + \beta)} (\theta e^{-\beta k} + \beta e^{-\theta k}).
\end{equation}
\end{theorem}

\subsection{Joint Asymptotic Inference via Z-Estimation}
To perform joint inference on $\Xi = (\theta, \beta, m)^{T}$, we frame the hybrid estimator as a simultaneous system of Z-estimating equations:
\begin{equation}\label{eq:z_sys}
\sum_{i=1}^{n} \Psi(X_{i}, \hat{\Xi}) = 0,
\end{equation}
where $\Psi(X_{i}, \Xi) = (\psi_{\theta}(X_{i}, \Xi), \psi_{\beta}(X_{i}, \Xi), \psi_{m}(X_{i}, \Xi))^{T}$ is defined by:
\begin{align}
\psi_{\theta}(X_{i}, \Xi) &= \frac{1}{\theta} - \frac{1}{\theta + \beta} - \mathbb{I}(X_{i} > m) \ln\left(\frac{X_{i}}{m}\right), \label{eq:psi_theta_def}\\
\psi_{\beta}(X_{i}, \Xi) &= \frac{1}{\beta} - \frac{1}{\theta + \beta} + \mathbb{I}(X_{i} \le m) \ln\left(\frac{X_{i}}{m}\right), \label{eq:psi_beta_def}\\
\psi_{m}(X_{i}, \Xi) &= \psi_{k}(\ln X_{i} - \ln m). \label{eq:psi_m_def}
\end{align}

\begin{theorem}[Joint Asymptotic Normality]\label{Joint-CAN}
Under the parameter-balancing condition \eqref{balance}, the joint hybrid estimator satisfies:
\[
\sqrt{n}(\hat{\Xi} - \Xi_{0}) \xrightarrow{d} \mathcal{N}_3(0, \Omega),
\]
where $\Omega = A^{-1} B (A^{-1})^{T}$ is the sandwich covariance matrix with $A = -\mathbb{E}[\nabla_{\Xi} \Psi(X, \Xi_{0})]$ and $B = \mathbb{E}[\Psi(X, \Xi_{0})\Psi(X, \Xi_{0})^{T}]$.
\end{theorem}

\begin{remark}
The simultaneous estimating equations \eqref{eq:z_sys} fall within Godambe's general framework of unbiased estimating functions, where $\Omega = A^{-1} B (A^{-1})^T$ corresponds to the inverse Godambe information matrix. We adopt the terminology of Z-estimation to accommodate the non-smooth Huber score $\psi_k(\cdot)$ and the profiled root-finding implementation.
\end{remark}

\subsection{Profiled Z-Estimation Procedure}\label{profZ}
Based on the above method of estimating  $ (\theta, \beta, m)^\top$ , we propose a \textit{Profiled Z-Estimator} that reduces the estimation to a stable one-dimensional root search over $m$ instead of multi-dimensional non-convex optimization over $\Xi = (\theta, \beta, m)^\top$. The method is summarized in the following steps:

\begin{enumerate}
    \item For any candidate $m > 0$, compute $T_1(m)$ and $T_2(m)$, evaluating the conditional shape MLEs $\hat{\theta}(m)$ and $\hat{\beta}(m)$ via \eqref{eq:mle_closed}.
    \item Update the dynamic tuning parameter via (\eqref{eq:k_approx}). That is,  $\hat{k}(m) = \frac{3}{\hat{\theta}(m) + \hat{\beta}(m)}$.
    \item Compute the scalar profiled estimating equation:
    \[
    S_n(m) = \sum_{i=1}^n \psi\left(\ln(X_i/m); \hat{\theta}(m), \hat{\beta}(m), \hat{k}(m)\right).
    \]
    \item Evaluate $S_n(m)$ over a candidate grid within the central quantile bracket $[q_{0.15}, q_{0.85}]$ to identify adjacent sign changes ($S_n(m_g) S_n(m_{g+1}) \le 0$). Brent's method (\texttt{uniroot}) is then applied to solve $S_n(\hat{m}) = 0$. If no sign change occurs, default to $\arg\min |S_n(m)|$.
    \item Obtain the final estimates via plug-in evaluation: $\hat{\theta} = \hat{\theta}(\hat{m})$ and $\hat{\beta} = \hat{\beta}(\hat{m})$.
\end{enumerate}

\section{Simulation Study}
We evaluate the finite-sample performance of $\hat{\Xi} = (\hat{\theta}, \hat{\beta}, \hat{m}_{M})^{T}$ across $M = 10{,}000$ Monte Carlo replications.
For each replication, synthetic samples of size $n \in \{50, 100, 500, 2500\}$ are generated by transforming independent uniform variates $U_i \sim \mathcal{U}(0, 1)$ through the exact quantile function \eqref{quantile}. The profiled Z-estimation procedure described in Section \ref{profZ} is applied to obtain the point estimates.

Theoretical asymptotic variances are computed via $[\Omega]_{rr}/n$, where $\Omega = A^{-1} B (A^{-1})^T$ is evaluated at $\Xi_0$. Empirical sandwich standard errors for data applications are obtained via plug-in evaluation:
\[
\widehat{\text{Cov}}(\hat{\Xi}) = \frac{1}{n} \hat{A}^{-1} \hat{B} (\hat{A}^{-1})^T = \frac{1}{n} \hat{\Omega}.
\]
As displayed in Table \ref{tab:simulation_results_1}, the empirical Monte Carlo variances converge cleanly to the theoretical sandwich variances as $n$ increases, confirming consistency and asymptotic normality.

\begin{table}[htbp]
\centering
\setlength{\tabcolsep}{3.5pt}
\caption{Monte Carlo simulation results based on 10,000 repetitions across sample sizes for various true parameter vectors $\Xi_{0} = (\theta, \beta, m)^{T}$.}
\label{tab:simulation_results_1}
\vspace{0.2cm}
\begin{tabular}{r ccc c ccc c ccc c ccc}
\toprule
 & \multicolumn{3}{c}{\textbf{Mean Estimate}} & & \multicolumn{3}{c}{\textbf{Mean Bias}} & & \multicolumn{3}{c}{\textbf{Emp. Variance}} & & \multicolumn{3}{c}{\textbf{Theor. Variance}} \\
\cline{2-4} \cline{6-8} \cline{10-12} \cline{14-16}
\vspace{-0.3cm} \\
$n$ & $\hat{\theta}$ & $\hat{\beta}$ & $\hat{m}_M$ & & $\hat{\theta}$ & $\hat{\beta}$ & $\hat{m}_M$ & & $\hat{\theta}$ & $\hat{\beta}$ & $\hat{m}_M$ & & $\hat{\theta}$ & $\hat{\beta}$ & $\hat{m}_M$ \\
\midrule
\multicolumn{16}{l}{\textbf{Case 1:} $\Xi_{0} = (0.8, 2.0, 8.0)^{T}$} \\
\midrule
50   & 0.8212 & 2.3834 & 7.9622 & &  0.0212 & 0.3834 & -0.0378 & & 0.0241 & 0.5598 & 23.7929 & & 0.0154 & 0.1800 & 0.8001 \\
100  & 0.8119 & 2.1981 & 7.8926 & &  0.0119 & 0.1981 & -0.1074 & & 0.0106 & 0.2799 &  2.2033 & & 0.0077 & 0.0900 & 0.4000 \\
500  & 0.8029 & 2.0274 & 8.0011 & &  0.0029 & 0.0274 &  0.0011 & & 0.0019 & 0.0332 &  0.2047 & & 0.0015 & 0.0180 & 0.0800 \\
2500 & 0.8006 & 2.0044 & 8.0010 & &  0.0006 & 0.0044 &  0.0010 & & 0.0004 & 0.0058 &  0.0346 & & 0.0003 & 0.0036 & 0.0160 \\
\midrule
\multicolumn{16}{l}{\textbf{Case 2:} $\Xi_{0} = (1.0, 3.0, 5.0)^{T}$} \\
\midrule
50   & 1.0388 & 3.4473 & 5.1825 & &  0.0388 & 0.4473 &  0.1825 & & 0.0396 & 1.0731 &  7.4888 & & 0.0233 & 0.4500 & 0.1667 \\
100  & 1.0158 & 3.3036 & 4.9553 & &  0.0158 & 0.3036 & -0.0447 & & 0.0155 & 0.6224 &  0.4230 & & 0.0117 & 0.2250 & 0.0833 \\
500  & 1.0035 & 3.0466 & 4.9989 & &  0.0035 & 0.0466 & -0.0011 & & 0.0028 & 0.0858 &  0.0422 & & 0.0023 & 0.0450 & 0.0167 \\
2500 & 1.0008 & 3.0072 & 5.0008 & &  0.0008 & 0.0072 &  0.0008 & & 0.0005 & 0.0151 &  0.0073 & & 0.0005 & 0.0090 & 0.0033 \\
\midrule
\multicolumn{16}{l}{\textbf{Case 3:} $\Xi_{0} = (1.5, 2.5, 1.8)^{T}$} \\
\midrule
50   & 1.5375 & 2.9838 & 1.7431 & &  0.0375 & 0.4838 & -0.0569 & & 0.1187 & 1.0615 &  0.0839 & & 0.0585 & 0.2292 & 0.0173 \\
100  & 1.5285 & 2.6823 & 1.7826 & &  0.0285 & 0.1823 & -0.0174 & & 0.0491 & 0.3368 &  0.0305 & & 0.0292 & 0.1146 & 0.0086 \\
500  & 1.5061 & 2.5263 & 1.7993 & &  0.0061 & 0.0263 & -0.0007 & & 0.0077 & 0.0379 &  0.0042 & & 0.0058 & 0.0229 & 0.0017 \\
2500 & 1.5012 & 2.5041 & 1.8000 & &  0.0012 & 0.0041 &  0.0000 & & 0.0014 & 0.0068 &  0.0007 & & 0.0012 & 0.0046 & 0.0003 \\
\midrule
\multicolumn{16}{l}{\textbf{Case 4:} $\Xi_{0} = (2.0, 1.5, 10.0)^{T}$} \\
\midrule
50   & 1.9734 & 1.7506 & 9.3096 & & -0.0266 & 0.2506 & -0.6904 & & 0.2621 & 0.3056 &  3.1187 & & 0.1333 & 0.0619 & 0.6667 \\
100  & 2.0041 & 1.5904 & 9.7298 & &  0.0041 & 0.0904 & -0.2702 & & 0.1064 & 0.0736 &  1.2234 & & 0.0667 & 0.0309 & 0.3334 \\
500  & 1.9997 & 1.5178 & 9.9423 & & -0.0003 & 0.0178 & -0.0577 & & 0.0219 & 0.0116 &  0.2473 & & 0.0133 & 0.0062 & 0.0667 \\
2500 & 1.9962 & 1.5051 & 9.9711 & & -0.0038 & 0.0051 & -0.0289 & & 0.0053 & 0.0023 &  0.0704 & & 0.0027 & 0.0012 & 0.0133 \\
\bottomrule
\end{tabular}
\end{table}

\section{Data Analysis}
We demonstrate the proposed model on the benchmark Danish Fire Insurance Claims dataset (\texttt{evir} package in \textsf{R}), consisting of $n = 2{,}167$ industrial fire claims (in millions of DKK, 1985 purchasing power). Descriptive statistics are provided in Table~\ref{tab:danish_summary}.
\begin{table}[htbp]
    \centering
    \caption{Summary Statistics for the Danish Fire Insurance Claims Dataset ($n = 2,167$, in Millions of DKK)}
    \label{tab:danish_summary}
    \begin{tabular}{lcccccc}
        \hline
        \textbf{Metric} & \textbf{Min.} & \textbf{1st Qu.} & \textbf{Median} & \textbf{Mean} & \textbf{3rd Qu.} & \textbf{Max.} \\
        \hline
        Claim Amount & $1.0000$ & $1.3210$ & $1.7740$ & $3.3850$ & $2.9670$ & $263.2500$ \\
        \hline
    \end{tabular}
\end{table}
Parameter estimates and asymptotic standard errors are reported in Table~\ref{tab:danish_estimation_results}.

\begin{table}[htbp]
\centering
\caption{Parameter estimates and asymptotic standard errors for the Danish Fire Insurance Claims dataset ($n = 2{,}167$, calculated tuning parameter $\hat{k} = 0.6725$).}
\label{tab:danish_estimation_results}
\vspace{0.2cm}

\begin{tabular}{lccc}
\toprule
\textbf{Parameter} & \textbf{Symbol} & \textbf{Estimate} & \textbf{Std. Error} \\
\midrule
Right-Tail Shape (Pareto) & $\hat{\theta}$ & 1.6892 & 0.0414 \\
Left-Body Shape (Power)    & $\hat{\beta}$  & 2.7721 & 0.0803 \\
Structural Threshold       & $\hat{m}_M$    & 1.7431 & 0.0211 \\
\bottomrule
\end{tabular}
\end{table}

The estimated threshold $\hat{m}_M = 1.7431$ million DKK divides the high-density body ($\hat{\beta} = 2.7721$) from the heavy Pareto upper tail ($\hat{\theta} = 1.6892$).
Figure \ref{diag} displays the fitted PDF and CDF against empirical claims.
\begin{figure}[htbp]
    \centering
    \includegraphics[width=1\linewidth]{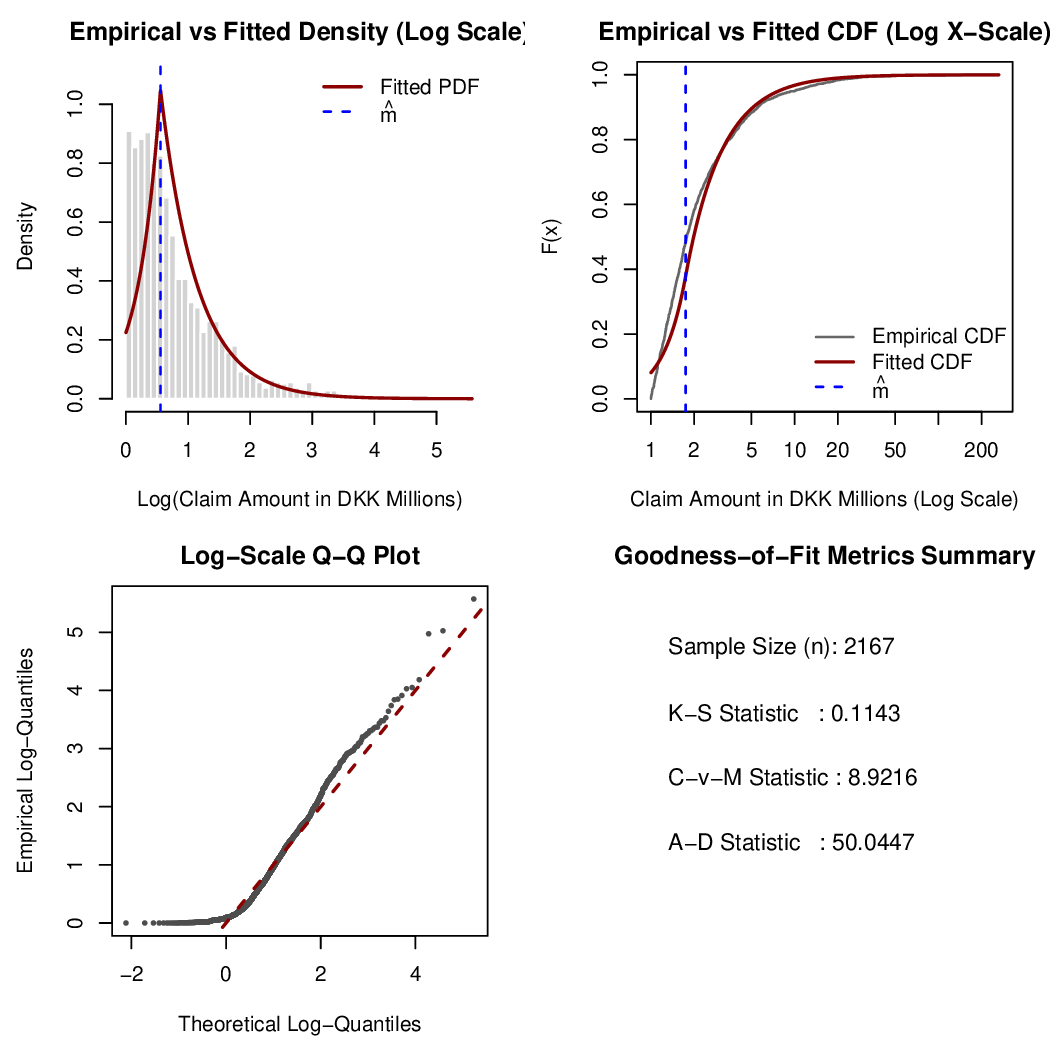}
    \caption{Diagnosis results for log-transformed Danish data}
    \label{diag}
\end{figure}

The log-transformed density plot highlights the density adaptation below $\hat{m}_M = 1.7431$, while the right-branch Pareto distribution captures extreme losses past 200 million DKK. Because $\hat{\theta} = 1.6892 \in (1, 2)$, the fitted model reflects a finite mean with infinite variance, in line with extreme risk modeling. Across the sample, distance-based goodness-of-fit statistics are: Kolmogorov-Smirnov (${\mathrm{K-S}\ Statistic} = 0.1143$), Cram\'er--von Mises ($C-v-M \ Statistic = 8.9216$), and Anderson--Darling ($A-D\ Statistic = 50.0447$).

\begin{remark}[Model Parsimony and Bounded Influence]
The Danish fire claims data feature an operational lower truncation boundary at $1.0$ million DKK. Although four-parameter composite formulations (such as composite lognormal--Pareto) add flexibility, our three-parameter formulation maintains high parsimony while regularizing cusp non-regularity. More importantly, while classical MLE is sensitive to catastrophic losses, the asymmetric Huber score clips extreme residuals beyond $[-k, k]$, preserving Fisher consistency via \eqref{balance} and yielding bounded gross-error sensitivity.
\end{remark}

\section{Conclusion}\label{sec:conclusion}
We have developed a unified inferential framework for a two-piece double Pareto-like distribution spliced at an unknown threshold $m$. Residual transformation to the log scale combined with an asymmetric Huber score resolves the non-regularity caused by the density cusp. The balancing condition $\theta(1-e^{-\beta k}) = \beta(1-e^{-\theta k})$ ensures Fisher consistency, and joint asymptotic normality is established with an explicit sandwich dispersion matrix. Profiled Z-estimation avoids multi-dimensional non-convexity by reducing the computation to a stable one-dimensional root search. Monte Carlo experiments confirmed the asymptotic theoretical dispersion, and the empirical application to the Danish Fire Claims demonstrated the practical efficacy of the model.

\appendix
\section*{Appendix: Proofs and Theoretical Derivations}
\renewcommand{\thesubsection}{A.\arabic{subsection}}

\subsection{Verification of Chanda's Regularity Conditions for \texorpdfstring{$\eta = (\theta, \beta)^{T}$}{eta = (theta, beta)}}
For a known or profiled structural threshold $m > 0$, consider the parameterized density family $f(x; \eta, m)$ on $x \in (0, \infty)$ with shape vector $\eta = (\theta, \beta)^T \in \Theta = (0, \infty) \times (0, \infty) \subset \mathbb{R}^2$. To establish the existence of a consistent root of the likelihood equations that is asymptotically normally distributed and corresponds to a local maximum of the likelihood function, we verify the multi-parameter Cram\'er-type regularity conditions formulated by Chanda (1954):\\

\textbf{Condition (0):}
\textbf{Identifiability of {$\eta = (\theta, \beta)^{T}:$}} For a fixed threshold $m > 0$, the parameter vector $\eta = (\theta, \beta)^T$ is identifiable if $f(x; \theta_1, \beta_1, m) = f(x; \theta_2, \beta_2, m)$ almost everywhere implies $\theta_1 = \theta_2$ and $\beta_1 = \beta_2$.  Suppose the two densities are equal for all $x \in (0, \infty)$.  
Considering the lower regime $0 < x \le m$ of the pdf (\ref{DPpdf}):$$\frac{\theta_1 \beta_1}{\theta_1 + \beta_1} \left(\frac{x}{m}\right)^{\beta_1 - 1} = \frac{\theta_2 \beta_2}{\theta_2 + \beta_2} \left(\frac{x}{m}\right)^{\beta_2 - 1}$$
Equating the exponents of $x$ forces $\beta_1 = \beta_2$. Consequently, the multiplicative constants must also match, yielding $\frac{\theta_1 \beta_1}{\theta_1 + \beta_1} = \frac{\theta_2 \beta_2}{\theta_2 + \beta_2}$.  
Considering the upper regime $x > m$:$$\frac{\theta_1 \beta_1}{\theta_1 + \beta_1} \left(\frac{m}{x}\right)^{\theta_1 + 1} = \frac{\theta_2 \beta_2}{\theta_2 + \beta_2} \left(\frac{m}{x}\right)^{\theta_2 + 1}$$
Equating the exponents of $x$ forces $\theta_1 = \theta_2$. Thus, $\theta_1 = \theta_2$ and $\beta_1 = \beta_2$, proving that distinct parameter vectors $\eta$ define distinct probability distributions.
Furthermore, the support $\mathcal{X} = (0, \infty)$ does not depend on $\eta$.\\

\textbf{Condition (i):} \textbf{Existence of Partial Derivatives in a Neighborhood of $\eta_0$:} \\
    For almost all $x \in (0, \infty)$ (specifically for all $x \ne m$), the log-likelihood derivative functions
      $$\frac{\partial^{i+j} \ln f(x; \eta, m)}{\partial \theta^i \partial \beta^j}, \quad \text{exists for } i, j \in \{0, 1, 2, 3\}\  \text{such that} \  1 \le i+j \le 3$$
  and for all $\eta$ belonging to an open neighborhood $U \subset \Theta$ containing the true parameter $\eta_0 = (\theta_0, \beta_0)^T$. \\
 
    The first partial derivatives of the log-density are:
    \begin{align}
    \frac{\partial \ln f(x; \eta, m)}{\partial \theta} &= \frac{1}{\theta} - \frac{1}{\theta + \beta} - \mathbb{I}(x > m)\ln\left(\frac{x}{m}\right), \label{eq:score_theta_app} \\
    \frac{\partial \ln f(x; \eta, m)}{\partial \beta} &= \frac{1}{\beta} - \frac{1}{\theta + \beta} + \mathbb{I}(x \le m)\ln\left(\frac{x}{m}\right). \label{eq:score_beta_app}
    \end{align}
    The second-order partial derivatives are:
    \begin{equation}\label{eq:second_deriv_app}
    \frac{\partial^2 \ln f}{\partial \theta^2} = -\frac{1}{\theta^2} + \frac{1}{(\theta+\beta)^2}, \quad 
    \frac{\partial^2 \ln f}{\partial \beta^2} = -\frac{1}{\beta^2} + \frac{1}{(\theta+\beta)^2}, \quad 
    \frac{\partial^2 \ln f}{\partial \theta \partial \beta} = \frac{1}{(\theta+\beta)^2},
    \end{equation}
    and the third-order partial derivatives are:
    \begin{equation}\label{eq:third_deriv_app}
    \frac{\partial^3 \ln f}{\partial \theta^3} = \frac{2}{\theta^3} - \frac{2}{(\theta+\beta)^3}, \quad 
    \frac{\partial^3 \ln f}{\partial \beta^3} = \frac{2}{\beta^3} - \frac{2}{(\theta+\beta)^3}, \quad 
    \frac{\partial^3 \ln f}{\partial \theta^2 \partial \beta} = \frac{\partial^3 \ln f}{\partial \theta \partial \beta^2} = -\frac{2}{(\theta+\beta)^3}.
    \end{equation}
    All partial derivatives \eqref{eq:score_theta_app}--\eqref{eq:third_deriv_app} exist and are continuous on $\Theta$. 

    \textbf{Condition (ii):} \textbf{Interchange of Differentiation and Integration and Information Equivalence:} \\
    The operations of integration over the support and partial differentiation with respect to $\eta_j$ can be interchanged for the first two derivatives. Integrating the first partial derivative with respect to $\theta$:
    \begin{align*}
    \int_{0}^{\infty} \frac{\partial f(x; \eta, m)}{\partial \theta} dx &= \int_{0}^{m} \frac{\partial f_{1}(x)}{\partial \theta} dx + \int_{m}^{\infty} \frac{\partial f_{2}(x)}{\partial \theta} dx \\
    &= \frac{\beta^{2}}{(\theta+\beta)^{2} m^{\beta}} \int_{0}^{m} x^{\beta-1} dx + \frac{\beta m^{\theta}}{\theta+\beta} \int_{m}^{\infty} x^{-(\theta+1)} \left[\frac{\beta}{\theta(\theta+\beta)} - \ln\left(\frac{x}{m}\right)\right] dx \\
    &= \frac{\beta^{2}}{(\theta+\beta)^{2}\beta} + \frac{\beta}{\theta+\beta}\left[\frac{\beta}{\theta(\theta+\beta)} \cdot \frac{1}{\theta} - \frac{1}{\theta^2}\right] \\
    &= \frac{\beta}{(\theta+\beta)^{2}} - \frac{\beta}{(\theta+\beta)^2} = 0.
    \end{align*}
    By symmetry, $\int_{0}^{\infty} \frac{\partial f(x; \eta, m)}{\partial \beta} dx = 0$. Hence, the score vector has zero mean:
    \[
    \mathbb{E}_\eta\left[\frac{\partial \ln f(X; \eta, m)}{\partial \eta_j}\right] = \int_0^\infty \frac{\partial f(x; \eta, m)}{\partial \eta_j} dx = 0 \quad \text{for } j \in \{1, 2\}.
    \]
    Similarly, differentiating $\int_0^\infty f(x; \eta, m) dx = 1$ twice under the integral sign gives $\int_0^\infty \frac{\partial^2 f(x; \eta, m)}{\partial \eta_j \partial \eta_k} dx = 0$, which yields the standard information equivalence identity:
    \[
    J_{jk}(\eta) = \mathbb{E}_\eta\left[ \frac{\partial \ln f(X; \eta, m)}{\partial \eta_j} \frac{\partial \ln f(X; \eta, m)}{\partial \eta_k} \right] = -\mathbb{E}_\eta\left[ \frac{\partial^2 \ln f(X; \eta, m)}{\partial \eta_j \partial \eta_k} \right] \quad \text{for } j, k \in \{1, 2\}.
    \]

    \textbf{Condition (iii):} \textbf{Positive Definiteness of the Fisher Information Matrix at $\eta_0$:} \\
    The $2 \times 2$ Fisher information matrix per observation $J(\eta) = (J_{jk}(\eta))_{2 \times 2}$ is:
    \[
    J(\eta) = \begin{pmatrix}
    \frac{\beta(\beta+2\theta)}{\theta^2(\theta+\beta)^2} & -\frac{1}{(\theta+\beta)^2} \\
    -\frac{1}{(\theta+\beta)^2} & \frac{\theta(\theta+2\beta)}{\beta^2(\theta+\beta)^2}
    \end{pmatrix}.
    \]
    The upper-left element is strictly positive:
    \[
    J_{11}(\eta) = \frac{\beta(\beta+2\theta)}{\theta^2(\theta+\beta)^2} > 0 \quad \text{for all } \theta > 0, \beta > 0.
    \]
    The determinant of $J(\eta)$ evaluates to:
    \[
    \det J(\eta) = J_{11} J_{22} - J_{12}^2 = \frac{\beta(\beta+2\theta)\theta(\theta+2\beta) - \theta^2\beta^2}{\theta^2\beta^2(\theta+\beta)^4} = \frac{2\theta\beta(\theta+\beta)^2}{\theta^2\beta^2(\theta+\beta)^4} = \frac{2}{\theta\beta(\theta+\beta)^2} > 0.
    \]
    Since $J_{11}(\eta_0) > 0$ and $\det J(\eta_0) > 0$, the Fisher information matrix $J(\eta_0)$ is strictly positive definite, satisfying Condition (iii) of Chanda (1954).\\

    \textbf{Condition (iv):} \textbf{Domination of the Third-Order Partial Derivatives:} \\
    From \eqref{eq:third_deriv_app}, 
    the third-order partial derivatives $\frac{\partial^3 \ln f(x; \eta, m)}{\partial \theta^{3-i} \partial \beta^i}$ for $i=0,1,2,3$ are deterministic constants that do not depend on $x$. Consequently, on any compact spherical neighborhood $U_0 \subset \Theta$ containing $\eta_0$, there exist dominating functions $M_{i}(x) \equiv M_{i}$ independent of $x$ such that:

\[
\left| \frac{\partial^3 \ln f(x; \eta, m)}{\partial \theta^{3-i} \partial \beta^i}\right|\le M_i  = \sup_{\eta \in U_0} \left| \frac{\partial^3 \ln f(x; \eta, m)}{\partial \theta^{3-i} \partial \beta^i} \right| < \infty \quad \text{for all } i \in \{0, 1, 2, 3\}.
\]
   
Since each $M_{i}$ is a finite constant, its expectation is trivially bounded:
\[
\mathbb{E}_{\eta_0}[M_{i}(X)] = M_{i} < \infty \quad \text{for all } i \in \{0, 1, 2, 3\}.
\]
  
Because Conditions (i)--(iv) of Chanda (1954) are satisfied, the main theorem of Chanda (1954) guarantees that, with probability tending to one as $n \to \infty$, there exists a sequence of roots $\hat{\eta}_{\text{MLE}} = (\hat{\theta}_{\text{MLE}}, \hat{\beta}_{\text{MLE}})^T$ of the likelihood equations that is consistent for $\eta_0$, provides a local maximum of the log-likelihood function, and satisfies:
\[
\sqrt{n}(\hat{\eta}_{\text{MLE}} - \eta_0) \xrightarrow{d} \mathcal{N}_2\big(0, \Sigma(\eta_{0}))\big),
\]
where  \[\Sigma(\eta_{0})=
J(\eta_{0})^{-1} =  \begin{pmatrix} \frac{\beta_{0}(\beta_{0} + 2\theta_{0})}{\theta_{0}^{2}(\theta_{0}+\beta_{0})^{2}} & -\frac{1}{(\theta_{0}+\beta_{0})^{2}} \\ -\frac{1}{(\theta_{0}+\beta_{0})^{2}} & \frac{\theta_{0}(\theta_{0} + 2\beta_{0})}{\beta_{0}^{2}(\theta_{0}+\beta_{0})^{2}} \end{pmatrix}^{-1} = \begin{pmatrix}
\frac{\theta_{0}^{3} + 2\theta_{0}^{2}\beta_{0}}{2\beta_{0}} & \frac{\theta_{0}\beta_{0}}{2} \\
\frac{\theta_{0}\beta_{0}}{2} & \frac{\beta_{0}^{3} + 2\beta_{0}^{2}\theta_{0}}{2\theta_{0}}
\end{pmatrix}.
\]

\subsection{Proof of Theorem \ref{k-value}}
\begin{proof}
Fisher consistency requires $\mathbb{E}_{\mu_{0}}[\psi_{k}(Z)] = 0$, where $Z = \ln(X/m_{0})$ has density $g(z) = C e^{\beta z}\mathbb{I}(z \le 0) + C e^{-\theta z}\mathbb{I}(z > 0)$ with $C = \frac{\theta\beta}{\theta+\beta}$. We evaluate $\mathbb{E}[\psi_{k}(Z)] = \sum_{j=1}^4 I_j$ using \eqref{psik}:
     $I_{1} = \int_{-\infty}^{-k} \beta C e^{\beta z} dz = C e^{-\beta k}$, \ \  
     $I_{2} = \int_{-k}^{0} \left(-\frac{\beta}{k}z\right) C e^{\beta z} dz = C \left[\frac{1}{\beta k} - e^{-\beta k} - \frac{e^{-\beta k}}{\beta k}\right]$.\\
     $I_{3} = \int_{0}^{k} \left(-\frac{\theta}{k}z\right) C e^{-\theta z} dz = C \left[e^{-\theta k} + \frac{e^{-\theta k}}{\theta k} - \frac{1}{\theta k}\right]$, \ \   
    $I_{4} = \int_{k}^{\infty} (-\theta) C e^{-\theta z} dz = -C e^{-\theta k}$. \\
Summing the pieces gives:
\[
\mathbb{E}[\psi_{k}(Z)] = \frac{C}{k} \left[\frac{1 - e^{-\beta k}}{\beta} - \frac{1 - e^{-\theta k}}{\theta}\right].
\]
Setting this expectation to zero directly yields \eqref{balance}.
\end{proof}

\begin{lemma}[Huber's Regularity Conditions]\label{Huber-conditions}
Let $Y_1, \dots, Y_n$ be i.i.d. random variables with density $g(y - \mu_0)$. Let $\hat{\mu}_n$ be an M-estimator obtained via $\psi$ as a solution of $\sum_{i=1}^n \psi(Y_i - \hat{\mu}_n) = 0$. If:
\begin{enumerate}
    \item $\psi(z)$ is non-increasing, continuous, and satisfies $\mathbb{E}_{\mu_0}[\psi(Y - \mu_0)] = 0$.
    \item $W = \mathbb{E}_{\mu_0}[\psi^2(Y - \mu_0)] < \infty$.
    \item $\lambda(\mu) = \mathbb{E}_{\mu_0}[\psi(Y - \mu)]$ is differentiable at $\mu_0$ with $V = -\lambda'(\mu_0) = -\mathbb{E}_{\mu_0}[\psi'(Y - \mu_0)] > 0$.
\end{enumerate}
Then $\sqrt{n}(\hat{\mu}_n - \mu_0) \xrightarrow{d} \mathcal{N}\left(0, \frac{W}{V^2}\right)$.
\end{lemma}

\subsection{Proof of Theorem \ref{mhat-CAN}}
\begin{proof}
We verify Conditions 1--3 of Lemma~\ref{Huber-conditions} for $\psi_k(z)$ from \eqref{psik}:
\begin{enumerate}
    \item $\psi_k(z)$ is continuous on $\mathbb{R}$ and non-increasing everywhere ($\psi_k'(z) \le 0$ almost everywhere). By Theorem~\ref{k-value}, condition \eqref{balance} guarantees Fisher consistency.
    \item Since $|\psi_k(z)| \le \max(\beta, \theta) < \infty$, the second moment $W = \mathbb{E}[\psi_k^2(Z)] < \infty$ exists. Direct integration yields \eqref{eq:W_def}.
    \item Differentiating the expected score yields $V = -\mathbb{E}[\psi_k'(Z)]$, which evaluates to \eqref{eq:V_def} and is strictly positive for all $k, \theta, \beta > 0$.
\end{enumerate}
A first-order Taylor expansion of $\frac{1}{n}\sum_{i=1}^n \psi_k(Y_i - \hat{\mu}_M) = 0$ around $\mu_0$ gives:
\begin{equation} \label{eq:app_taylor_m}
0 = \frac{1}{\sqrt{n}}\sum_{i=1}^n \psi_k(Y_i - \hat{\mu}_M) = \frac{1}{\sqrt{n}}\sum_{i=1}^n \psi_k(Y_i - \mu_0) + \left[ \frac{1}{n}\sum_{i=1}^n \psi'_k(Y_i - \bar{\mu}) \right] \sqrt{n}(\hat{\mu}_M - \mu_0),
\end{equation}
where $\bar{\mu}$ lies on the line segment joining $\hat{\mu}_M$ and $\mu_0$. 

By the Weak Law of Large Numbers (WLLN) and consistency of $\hat{\mu}_M$, the bracketed term converges in probability to its expectation:
\[
\frac{1}{n}\sum_{i=1}^n \psi'_k(Y_i - \bar{\mu}) \xrightarrow{p} E_{\mu_0}[\psi'_k(Y - \mu_0)] \equiv -V.
\]
Evaluating $V = -E[\psi'_k(Z)]$ directly using density $g(z)$ yields \eqref{eq:V_def}. Rearranging \eqref{eq:app_taylor_m} leads to the linear asymptotic representation: 
\[
\sqrt{n}(\hat{\mu}_M - \mu_0) = \frac{\frac{1}{\sqrt{n}} \sum_{i=1}^n \psi_k(Y_i - \mu_0)}{V} + o_p(1) 
\]
By the classical central limit theorem, $\frac{1}{\sqrt{n}} \sum_{i=1}^n \psi_k(Y_i - \mu_0) \xrightarrow{d} \mathcal{N}(0,W)$, where $W$ is given by \eqref{eq:W_def}. Hence $\sqrt{n}(\hat{\mu}_M - \mu_0) \xrightarrow{d} \mathcal{N}\left(0, \frac{W}{V^2}\right)$.
Applying the Delta method with $h(\mu) = e^\mu$ ($h'(\mu_0) = m_0$) yields $\sqrt{n}(\hat{m}_M - m_0) \xrightarrow{d} \mathcal{N}(0, \sigma_M^2)$ with $\sigma_M^2 = m_0^2 \frac{W}{V^2}$.
\end{proof}

\subsection{Proof of Theorem \ref{Joint-CAN} and Sandwich Covariance Derivation}
\begin{proof}
Let $\Psi(X_i, \Xi) = (\psi_\theta(X_i, \Xi), \psi_\beta(X_i, \Xi), \psi_m(X_i, \Xi))^T$ be the estimating functions defined in \eqref{eq:psi_theta_def}--\eqref{eq:psi_m_def}. For any non-zero vector $t \in \mathbb{R}^3$, the scalar projections $\{S_i = t^T \Psi(X_i, \Xi_0)\}$ forms a sequence of iid rvs with $\mathbb{E}[S_i] = 0$ and $\text{Var}(S_i) = t^T \mathbf{B} t < \infty$ with $\mathbf{B} = \mathbb{E}[\Psi(X, \Xi_0)\Psi(X, \Xi_0)^T]$. By the Lindeberg--L\'evy CLT,
\begin{equation*}
\frac{1}{\sqrt{n}}\sum_{i=1}^n{S_i} \xrightarrow{d} \mathcal{N}(0, \mathbf{t^TBt})
\end{equation*}
and by Cram\'er--Wold theorem, $\frac{1}{\sqrt{n}}\sum_{i=1}^n \Psi(X_i, \Xi_0) \xrightarrow{d} \mathcal{N}_3(0, \mathbf{B})$.
Next, expanding $\sum_{i=1}^n \mathbf{\Psi}(X_i, \hat{\boldsymbol{\Xi}}) = \mathbf{0}$ via a multivariate first-order Taylor expansion around $\boldsymbol{\Xi}_0$:
\[
\mathbf{0} = \frac{1}{\sqrt{n}}\sum_{i=1}^n \mathbf{\Psi}(X_i, \boldsymbol{\Xi}_0) + \left[ \frac{1}{n}\sum_{i=1}^n \nabla_{\boldsymbol{\Xi}} \mathbf{\Psi}(X_i, \bar{\boldsymbol{\Xi}}) \right] \sqrt{n}(\hat{\boldsymbol{\Xi}} - \boldsymbol{\Xi}_0),
\]
where $\bar{\boldsymbol{\Xi}}$ lies on the line segment connecting $\hat{\boldsymbol{\Xi}}$ and $\boldsymbol{\Xi}_0$. By the Uniform Law of Large Numbers, $\frac{1}{n}\sum_{i=1}^n \nabla_{\boldsymbol{\Xi}} \mathbf{\Psi}(X_i, \bar{\boldsymbol{\Xi}}) \xrightarrow{p} -\mathbf{A}$, where $\mathbf{A} = -E[\nabla_{\boldsymbol{\Xi}} \mathbf{\Psi}(X, \boldsymbol{\Xi}_0)]$.

Rearranging and inverting the non-singular matrix $\mathbf{A}$:
\[
\sqrt{n}(\hat{\boldsymbol{\Xi}} - \boldsymbol{\Xi}_0) = \mathbf{A}^{-1} \left[ \frac{1}{\sqrt{n}}\sum_{i=1}^n \mathbf{\Psi}(X_i, \boldsymbol{\Xi}_0) \right] + o_p(1).
\]
Applying Slutsky's Theorem yields the final sandwich variance result:
\[
\sqrt{n}(\hat{\boldsymbol{\Xi}} - \boldsymbol{\Xi}_0) \xrightarrow{d} \mathcal{N}\left(\mathbf{0}, \, \mathbf{\Omega}\right), \text{where} \  \mathbf{\Omega} = \mathbf{A}^{-1} \mathbf{B} (\mathbf{A}^{-1})^T. 
\]

The matrices $\mathbf{A}$ and $\mathbf{B}$ are evaluated as: 

\[ 
A = \begin{pmatrix}
J_{11} & J_{12} & -\frac{\beta_0}{m_0(\theta_0+\beta_0)} \\
J_{21} & J_{22} & \frac{\theta_0}{m_0(\theta_0+\beta_0)} \\
0 & 0 & \frac{V}{m_0}
\end{pmatrix} 
\quad \text{and} \quad 
B = \begin{pmatrix}
J_{11} & J_{12} & B_{13} \\
J_{21} & J_{22} & B_{23} \\
B_{31} & B_{32} & W
\end{pmatrix}
\]
where $J_{ij}$ are the elements of $J(\eta_0)$ from \eqref{eq:fisher_matrix}, $V$ is given in \eqref{eq:V_def} and $W$ defined in \eqref{eq:W_def}. The cross-covariance elements $B_{13} =\mathbb{E}(\psi_{\theta}\psi_m)$ and $B_{23} = \mathbb{E}(\psi_{\beta}\psi_m)$ are computed by integration over $Z=ln(X/m_0)$ with density $g(z) = Ce^{\beta_0z}\mathbb{I}(z \le 0)+Ce^{-\theta_0z}\mathbb{I}(z > 0)$ and noting that $\mathbb{E}(\psi_k(Z))=0$:
\begin{align*}
B_{13} =\mathbb{E}(\psi_{\theta}\psi_m)= \mathbb{E}\left[\left( \frac{\beta_0}{\theta_0(\theta_0+\beta_0)}-\mathbb{I}(z > 0)z\right) \psi_k(z)\right] = \mathbb{E}[\mathbb{I}(z > 0)z\psi_k(z)]\\
= -\frac{\beta_0}{\theta_0+\beta_0} \left[ \frac{2 - 2e^{-\theta_0 k} - \theta_0 k e^{-\theta_0 k}}{\theta_0 k} \right], \\ 
B_{23} = \mathbb{E}(\psi_{\beta}\psi_m)=\mathbb{E}\left[\left( -\frac{\theta_0}{\beta_0(\theta_0+\beta_0)}-\mathbb{I}(z \le 0)z\right) \psi_k(z)\right] = \mathbb{E}[\mathbb{I}(z \le 0)z\psi_k(z)]\\
=\frac{\theta_0}{\theta_0+\beta_0} \left[ \frac{2 - 2e^{-\beta_0 k} - \beta_0 k e^{-\beta_0 k}}{\beta_0 k} \right].
\end{align*}

Inverting $A$ gives:
\[
A^{-1} = \begin{pmatrix}
\Sigma_{11} & \Sigma_{12} & -\frac{m_0}{V} C_1 \\
\Sigma_{21} & \Sigma_{22} & -\frac{m_0}{V} C_2 \\
0 & 0 & \frac{m_0}{V}
\end{pmatrix},
\]
where $\Sigma = J^{-1}$ is given in \eqref{eq:sigma_eta} and $C_i = \Sigma_{i1} A_{13} + \Sigma_{i2} A_{23}$ for $i \in \{1, 2\}$. 

Evaluating $\Omega = A^{-1} B (A^{-1})^T$ yields:
\begin{align*}
\Omega_{33} &= m_0^2 \frac{W}{V^2} = \sigma_M^2, \\
\Omega_{ij} &= \Sigma(\eta_0)_{ij} + \frac{m_0^2}{V^2}\left( C_i C_j W - V(C_i D_j + C_j D_i) \right), \quad i,j \in \{1,2\}, \\
\Omega_{i3} &= \frac{m_0}{V}\left( D_i - C_i \frac{W}{V} \right), \quad i \in \{1,2\},
\end{align*}
where $D_i = \Sigma_{i1} B_{13} + \Sigma_{i2} B_{23}$. This completes the derivation.
\end{proof}

\end{document}